\documentclass{llncs}
\usepackage{amssymb,amsmath}
\usepackage{enumerate}
\usepackage{graphicx}

  \newcommand{\Alp}{\textsf{Alph}}
  \newcommand{\proot}{\textsf{root}}
  \newcommand{\per}{\textsf{per}}
  \newcommand{\runs}{\textsf{runs}}
  \newcommand{\hpruns}{\textsf{hp-runs}}

  \newcommand{\nth}{\mbox{${}^{\textsl{\scriptsize th}}$}}
  \newcommand{\nd}{\mbox{${}^{\textsl{\scriptsize nd}}$}}

  \def\rdots{\mathinner{\ldotp\ldotp}}

  \date{}
  \author{\bf
    Maxime Crochemore\inst{1}\fnmsep\inst{3}
    \and
    Costas Iliopoulos\inst{1}\fnmsep\inst{4}
    \and
    Marcin Kubica\inst{2}
    \and \\
    Jakub Radoszewski\inst{2}
    \and
    Wojciech Rytter%
\thanks{Supported by grant N206 004 32/0806 of the Polish Ministry
 of Science and Higher Education.}
    \inst{2}\fnmsep\inst{5}
    \and
    Tomasz Wale\'n\inst{2}
  }

  \institute{
    Dept.~of Computer Science, King's College London, London WC2R 2LS, UK \\
    \email{[maxime.crochemore,csi]@kcl.ac.uk}
    \and
    Dept.~of Mathematics, Computer Science and Mechanics, \\
    University of Warsaw, Warsaw, Poland\\
    \email{[kubica,jrad,rytter,walen]@mimuw.edu.pl}
    \and
    Universit\'e Paris-Est, France
    \and
    Digital Ecosystems \& Business Intelligence Institute, \\
    Curtin University of Technology, Perth WA 6845, Australia
    \and
    Dept. of Math. and Informatics,\\
    Copernicus University, Toru\'n, Poland
  }

  \title{
      On the maximal number of \\ highly periodic runs in a string%
\thanks{Research supported in part by the Royal Society, UK.}
  }

\begin{document}
  \maketitle
\begin{abstract}
A run is a maximal occurrence of a repetition $v$ with a period $p$ such that
$2p \le |v|$.
The maximal number of runs in a string of length $n$ was studied by several
authors and it is known to be between $0.944 n$
and $1.029 n$.
We investigate highly periodic runs, in which the shortest period $p$
satisfies $3p \le |v|$.
We show the upper bound $0.5n$ on the maximal number of such runs in a string
of length $n$ and construct a sequence of words for which we obtain the lower
bound $0.406 n$.
\end{abstract}

\section{Introduction}
Repetitions and periodicities in strings are one of the fundamental topics in
combinatorics on words \cite{Karhumaki,Lothaire}.
They are also important in other areas: lossless compression, word
representation, computational biology etc.
Repetitions are studied from different directions: classification of
words not containing repetitions of a given exponent, efficient identification
of factors being repetitions of different types and finally computing the
bounds of the number of repetitions of a given exponent that a string may
contain, which we consider in this paper.
Both the known results in the topic and a deeper description of the
motivation can be found in the survey by Crochemore et al. \cite{Survey}.

The concept of runs (also called maximal repetitions) has been introduced to
represent all repetitions in a string in a succinct manner.
The crucial property of runs is that their maximal number in a string of
length $n$ (denoted as $\runs(n)$) is $O(n)$ \cite{KolpakovKucherov}.
Due to the work of many people, much better bounds on $\runs(n)$ have been
obtained.
The lower bound $0.927 n$ was first proved in \cite{FranekYang}.
Afterwards it was improved by Kusano et al.~\cite{Matsubara} to
$0.944 n$ employing computer experiments and very recently by Simpson
\cite{Simpson} to $0.944575712 n$.
On the other hand, the first explicit upper bound $5n$ was settled in
\cite{Rytter06}, afterwards it was systematically improved to
$3.44 n$ \cite{Rytter07},
$1.6 n$ \cite{CrochemoreIlie,CrochemoreIlie08} and
$1.52 n$ \cite{Giraud08}.
The best known result $\runs(n) \le 1.029 n$ is due to Crochemore et
al.~\cite{DBLP:conf/cpm/CrochemoreIT08}, but it is conjectured
\cite{KolpakovKucherov} that $\runs(n)<n$.
The maximal number of runs was also studied for special types of strings
and tight bounds were established for Fibonacci strings
\cite{KolpakovKucherov,RytterFib}
and more generally Sturmian strings \cite{BaturoPiatkowskiRytter}.

The combinatorial analysis of runs in strings is strongly related to the
problem of estimation of the maximal number of occurrences of squares in
a string.
In the latter the gap between the upper and lower bound is much larger than
for runs \cite{Survey,CrochemoreR95}.
However, a recent paper \cite{Iwoca} by some of the authors shows that
introduction of exponents larger than 2 can lead to obtaining tighter bounds
for the number of corresponding occurrences.

In this paper we introduce and study the concept of highly periodic runs
(hp-runs) in which the period is at least three times shorter than the run.
We show the following bounds on the number $\hpruns(n)$ of such runs
in a string of length $n$:
$$0.406 n\le \hpruns(n)\le \frac{n-1}2$$
The upper bound is achieved by analyzing prime words (i.e.~words that are
primitive and minimal/maximal in the class of their cyclic equivalents)
that appear as periods of hp-runs.
As for the lower bound, we give a simple argument that leads to $0.4 n$ bound
and then describe a family of words that improves this bound to $0.406 n$.

\section{Definitions}
We consider \emph{words} over a finite alphabet $A$, $u\in A^*$; 
by $\varepsilon$ we denote an empty word;
the positions in a word $u$ are numbered from $1$ to $|u|$.
By $\Alp(u)$ we denote the set of all letters of $u$.
For $u=u_1u_2\ldots u_m$, by $u[i\rdots j]$ we denote a \textit{factor}
of $u$ equal to $u_i\ldots u_j$ (in particular $u[i]=u[i\rdots i]$).
Words $u[1\rdots i]$ are called prefixes of $u$, and words $u[i\rdots m]$
--- suffixes of $u$.
We say that positive integer $p$ is the (shortest) \emph{period} of a word
$u=u_1\ldots u_m$ (notation: $p=\per(u)$) if $p$ is the smallest number
such that $u_i=u_{i+p}$ holds for all $1\le i\le m-p$.

If $w^k=u$ ($k$ is a non-negative integer) then we say that $u$
is the $k\nth$ power of the word $w$.
A \emph{square} is the $2\nd$ power of some word.
The \emph{primitive root} of a word $u$, denoted $\proot(u)$, is
the shortest such word $w$ that $w^k=u$ for some positive $k$.
We call a word $u$ \emph{primitive} if $\proot(u)=u$, otherwise 
it is called \emph{nonprimitive}.
We say that words $u$ and $v$ are cyclically equivalent (or that one of them
is a cyclic rotation of the other) if $u=xy$ and $v=yx$ for some $x,y\in A^*$.
It is a simple observation that if $u$ and $v$ are cyclically equivalent
then $\proot(u)=\proot(v)$.

Let us assume that $A$ is totally ordered by $\le$ what induces
a lexicographical order in $A^*$, also denoted by $\le$.
We say that $u\in A^*$ is a \emph{prime word} if it is primitive and minimal
or maximal in the class of words that are cyclically equivalent to it.
It can be proved \cite{Lothaire} that a prime word $u$ cannot have a proper
(i.e.~non-empty and different than $u$) prefix that would also be its
suffix.

A \emph{run} (also called a maximal repetition) in a string $u$ is an interval
$[i\rdots j]$ such that both the associated factor $u[i\rdots j]$ has period
$p$, $2p \le j-i+1$, and the property cannot be extended to the right nor to
the left: $u[i-1] \ne u[i+p-1]$ and $u[j-p+1] \ne u[j+1]$ when the letters are
defined.
A \emph{highly periodic run} (hp-run) is a run $[i\rdots j]$ for which the
shortest period $p$ satisfies $3p\le j-i+1$.
For simplicity, in the further text we sometimes refer to runs or hp-runs as
to occurrences of corresponding factors of $u$.

\section{Upper bound}
Let $u\in A^*$ be a word of length $n$.
By $P=\{p_1,p_2,\ldots,p_{n-1}\}$ we denote the set of inter-positions of $u$
that are located \emph{between} pairs of consecutive letters of $u$.

We define a function $F$ that assigns to each hp-run $v$ in a string the set
of \emph{handles} among all inter-positions within $v$.
Hence, $F$ is a mapping from the set of hp-runs occurring in $u$ to the set
$2^P$ of subsets of $P$.
Let $v$ be a hp-run with period $p$ and let $w$ be the prefix of $v$ of length $p$.
By $w_{min}$ and $w_{max}$ we denote words cyclically equivalent to $w$ that
are minimal and maximal in lexicographical order.
We define $F(v)$ as follows:
\begin{enumerate}[a)]
  \item if $w_{min}\ne w_{max}$ then $F(v)$ contains inter-positions between
  consecutive occurrences of $w_{min}$ and between consecutive
  occurrences of $w_{max}$ within $v$
  \item if $w_{min}=w_{max}$ then $F(v)$ contains all inter-positions
  within $v$.
\end{enumerate}

\begin{lemma}\label{l:wminwmax_prime}
$w_{min}$ and $w_{max}$ are prime words.
\end{lemma}

\begin{proof}
By the definition of $w_{min}$ and $w_{max}$, it suffices to show that both
words are primitive.
This follows from the fact that, due to the minimality of $p$, $w$ is
primitive and that $w_{min}$ and $w_{max}$ are cyclically equivalent to $w$.
\qed
\end{proof}

\begin{lemma}\label{l:wminwmax_caseb}
Case b) from the above definition implies that $|w_{min}|=1$.
\end{lemma}

\begin{proof}
$w_{min}$ is primitive, therefore if $|w_{min}|\ge 2$ then $w_{min}$ would
contain at least two distinct letters, $a=w_{min}[1]$ and $b=w_{min}[i]\ne a$.
If $b<a$ ($b>a$) then the cyclic rotation of $w_{min}$ by $i-1$ letters would
be lexicographically smaller (greater) than $w_{min}$ --- a contradiction.
\qed
\end{proof}

\noindent
Note that in case b) of the definition of $F$ obviously $F(v)$ contains at
least two distinct handles.
The following lemma concludes that the same property also holds in case a).

    \begin{figure}[th]
      \begin{center}
        \includegraphics[width=7cm]{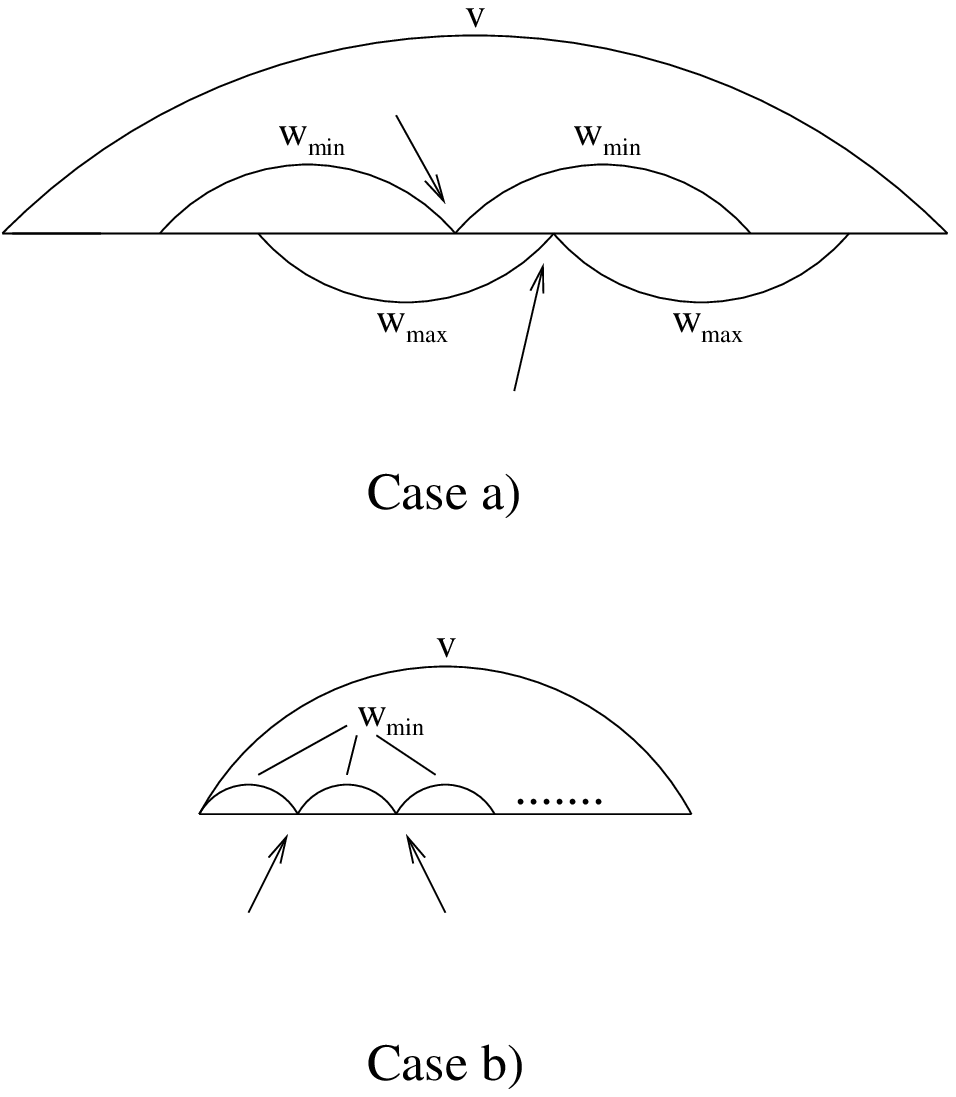}
        \caption{\label{f:wminwmax}
          Illustration of the definition of $F$ and Lemma \ref{l:fdef}.
          The arrows in the figure point to positions from the set of handles
          $F(v)$.
        }
      \end{center}
    \end{figure}

\begin{lemma}\label{l:fdef}
Each of the words $w_{min}^2$ and $w_{max}^2$ is a factor of $v$.
\end{lemma}

\begin{proof}
Recall that $3p\le |v|$, where $p=\per(v)$.
By Lemma \ref{l:wminwmax_caseb}, this concludes the proof in case b).
As for the proof in case a), it suffices to note that
the first occurrences of each of the words $w_{min}$, $w_{max}$
within $v$ start non-further than $p$ positions from the beginning of $v$.
\qed
\end{proof}

\noindent
We now show a crucial property of $F$.

\begin{lemma}\label{l:pairs_of_runs}
$F(v_1) \cap F(v_2)=\emptyset$ for every two distinct hp-runs $v_1,v_2$ in
$u$.
\end{lemma}

\begin{proof}
Assume to the contrary that $p_i \in F(v_1) \cap F(v_2)$ is a handle of two
different runs $v_1$ and $v_2$.
By Lemmas \ref{l:wminwmax_prime} and \ref{l:fdef}, $p_i$ is located in the
middle of two squares $w_1^2$ and $w_2^2$ of prime words, where
$|w_1|=\per(v_1)$ and $|w_2|=\per(v_2)$.
$w_1\ne w_2$, since in the opposite cases runs $v_1$ and $v_2$ would be the
same.
W.l.o.g.~assume that $|w_1|<|w_2|$.
Then, word $w_1$ is both a prefix and a suffix of $w_2$
(see fig.~\ref{f:prefsuf}), what contradicts the primality of $w_2$.
\qed

     \begin{figure}[th]
      \begin{center}
        \includegraphics[width=7.5cm]{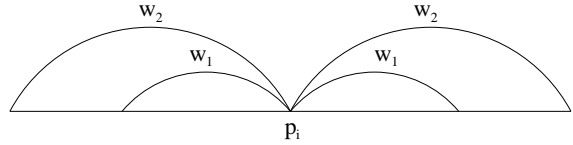}
        \caption{\label{f:prefsuf}
          A situation where $p_i$ is in the middle of two different squares
          $w_1^2$ and $w_2^2$.
        }
      \end{center}
    \end{figure}
\end{proof}

\noindent
The following theorem concludes the analysis of the upper bound.

\begin{theorem}
A word $u\in A^*$ of length $n$ may contain at most $\frac{n-1}2$ runs.
\end{theorem}

\begin{proof}
Due to Lemma \ref{l:fdef}, for each hp-run $v$ within $u$, $|F(v)| \ge 2$.
Since $|P|=n-1$, Lemma \ref{l:pairs_of_runs} implies the conclusion of the
theorem.
\qed
\end{proof}

\section{Lower bound}

\begin{lemma}\label{l:improvement}
Let $s$ be a word and denote:
$$r = \hpruns(s), \quad \ell = |s|$$
There exists a sequence of words $(s_n)_{n=0}^{\infty}$, $s_0=s$, such that
$$r_n = \hpruns(s_n), \quad \ell_n = |s_n| \quad\mbox{and}
  \quad \lim_{n\rightarrow\infty} \frac{r_n}{\ell_n} = \frac{r}\ell+\frac{1}{5\ell}$$
\end{lemma}

\begin{proof}
We define the sequence $s_n$ recursively.
Denote $A=\Alp(s_n)$ and let $\overline{A}$ be a disjoint copy of $A$.
By $\overline{s_n}$ we denote the word obtained from $s_n$ by substituting
letters from $A$ with the corresponding letters from $\overline{A}$.
We define $s_{n+1} = (s_n\overline{s_n})^3$.

Recall that $\ell_0=\ell$, $r_0=r$ and note that for $n\ge 1$
$$\ell_n=6\ell_{n-1}, \quad r_n=6r_{n-1}+1$$
By simple induction this concludes that
$$\frac{r_n}{\ell_n} \quad=\quad
  \frac{r}\ell + \frac{1}{\ell} \sum_{i=1}^n \frac{1}{6^i} \quad=\quad
  \frac{r}\ell + \frac{1}{5\ell}\Bigl(1-\frac{1}{6^{n+1}}\Bigr)$$
Taking $n\rightarrow\infty$ in the above formula we obtain the conclusion
of the lemma.
\qed
\end{proof}

\noindent
Starting with the 3-letter word $s=a^3$ for which $r/\ell=1/3$, from Lemma
\ref{l:improvement} we obtain the bound $0.4 n$.
This bound is, however, not optimal --- we will show an example of a sequence
of words for which we obtain the bound $0.406 n$.

Let $A=\{a,b\}$.
We denote:
  $$X=\left(a^3b^3\right)^3, \quad
  Y=a^4b^3a, \quad
  \alpha=XY, \quad
  \beta=Xa$$

\begin{lemma}\label{l:props_alphabeta}
A couple of important properties of words $\alpha$ and $\beta$:
\begin{itemize}
  \item $XYX$ introduces a new hp-run with the period 7.
  Hence, each of the pairs $\alpha\alpha$ and $\alpha\beta$ introduces a new
  hp-run.
  \item $\beta$ is a prefix of $\alpha$.
  Hence, $\alpha\beta\alpha\beta\alpha\alpha$ introduces the hp-run
  $(\alpha\beta)^3$.
  \item $Y$ is a prefix of $aX$, therefore $\alpha$ is a prefix of
  $\beta\alpha$.
  Hence, $\alpha\alpha\beta\alpha$ introduces the hp-run $\alpha^3$.
\end{itemize}
\end{lemma}

\noindent
Now we will also be dealing with a new alphabet $A'=\{\alpha,\beta\}$.
We define the Fibonacci morphism $h$ as:
$$h(\alpha)=\alpha\beta, \quad h(\beta)=\alpha$$
Let
$$f_n=h^n(\alpha), \quad r_n = \hpruns(f_n), \quad \ell_n = |f_n|$$

\begin{center}
\begin{tabular}{|r|r|r|r|p{5.5cm}|}
\hline
\ \ $n$ & \ \ \ \ $r_n$ & \ \ \ \ $\ell_n$ & \ \ \ \ \ \ $r_n/\ell_n$ & $f_n$ \\\hline\hline
0 & $9$ & $26$ & $0.3462$ & $\alpha$ \\\hline
1 & $17$ & $45$ & $0.3778$ & $\alpha\beta$ \\\hline
2 & $26$ & $71$ & $0.3662$ & $\alpha\beta\alpha$ \\\hline
3 & $45$ & $116$ & $0.3879$ & $\alpha\beta\alpha\alpha\beta$ \\\hline
4 & $71$ & $187$ & $0.3796$ & $\alpha\beta\alpha\alpha\beta\alpha\beta\alpha$ \\\hline
5 & $119$ & $303$ & $0.3927$ & $\alpha\beta\alpha\alpha\beta\alpha\beta\alpha\alpha\beta\alpha\alpha\beta$ \\\hline
6 & $192$ & $490$ & $0.3918$ &
  $\alpha\beta\alpha\alpha\beta\alpha\beta\alpha\alpha\beta\alpha\alpha\beta\alpha\beta\alpha\alpha\beta\alpha\beta\alpha$ \\\hline
\end{tabular}

\medskip

Table 1: A first few words of the sequence $f_n$ with the corresponding \\ terms of sequences $r_n$ and $\ell_n$.
\end{center}

\vskip 1cm
\begin{theorem}
$$\lim_{n\rightarrow\infty} \frac{r_n}{\ell_n} > 0.406$$
In particular,
$$\frac{r_{19}}{\ell_{19}} \ge \frac{103\,664}{255\,329} > 0.406$$
\end{theorem}
\begin{proof}
We start with the values $\ell_n,r_n$ for $n\le 4$ that are precomputed in
Table 1 and show that for $n\ge 5$ the following recursive formulas hold:
\begin{eqnarray}
\label{e:ln}
  \ell_n & = & \ell_{n-1}+\ell_{n-2} \\
\label{e:rn_even}
  r_n & \ge & r_{n-1}+r_{n-2}+n-4 \quad\quad\mbox{\it if}\quad 2 \mid n \\
\label{e:rn_odd}
  r_n & \ge & r_{n-1}+r_{n-2}+n-2 \quad\quad\mbox{\it if}\quad 2 \nmid n
\end{eqnarray}
The ``in particular'' part of the lemma is a straightforward consequence of
the formulas.

\eqref{e:ln} is obvious, therefore we concentrate on the inequalities for
$r_n$.
The recursive part of each of them ($r_{n-1}+r_{n-2}$) is a consequence of
the formula $f_n=f_{n-1}f_{n-2}$ and the fact that Fibonacci words contain
repetitions of exponent at most $2+\Phi < 4$, see \cite{MignosiPirillo}.
Due to Lemma \ref{l:props_alphabeta}, for even values of $n$ a new hp-run is
introduced upon concatenation --- see the example for $n=6$:
$$\alpha\beta\alpha\alpha\beta\alpha\beta\alpha\alpha\beta\alpha\underbrace{\alpha\beta | \alpha\beta\alpha\alpha}_{}\beta\alpha\beta\alpha$$
and for odd values of $n$, three more hp-runs appear, as in the following
example for $n=5$:
$$\alpha\beta\alpha\alpha\beta\alpha\beta\underbrace{\alpha | \alpha}_{}\beta\alpha\alpha\beta$$
$$\alpha\beta\alpha\alpha\beta\alpha\beta\underbrace{\alpha | \alpha\beta\alpha}_{}\alpha\beta$$
$$\alpha\beta\alpha\underbrace{\alpha\beta\alpha\beta\alpha | \alpha}_{}\beta\alpha\alpha\beta$$
Apart from that, since
$$h(\alpha\beta\alpha\beta\alpha\alpha) = \underbrace{\alpha\beta\alpha\alpha\beta\alpha\alpha\beta\alpha}_{}\beta$$
contains a hp-run $f_2^3$, word $f_n$ introduces $n-5$ new hp-runs composed
form $f_2^3, f_3^3, \ldots, f_{n-4}^3$, each created by iterating
$h^i(\alpha\beta\alpha\beta\alpha\alpha)$ --- see the example for $n=7$:
$$\alpha\beta\alpha\alpha\beta\alpha\beta\alpha\alpha\beta\alpha\alpha\beta\alpha\beta\alpha\alpha\beta\underbrace{\alpha\beta\alpha |
  \alpha\beta\alpha\alpha\beta\alpha}_{}\beta\alpha\alpha\beta\alpha\alpha\beta$$  
$$\alpha\beta\alpha\alpha\beta\alpha\beta\alpha\underbrace{\alpha\beta\alpha\alpha\beta\alpha\beta\alpha\alpha\beta\alpha\beta\alpha |
  \alpha\beta}_{}\alpha\alpha\beta\alpha\beta\alpha\alpha\beta\alpha\alpha\beta$$  
In total, we obtain $n-4$ new hp-runs for even $n$ and $n-2$ for odd $n$, what
concludes the proof of the inequalities.
\qed
\end{proof}

\bibliographystyle{abbrv}
\bibliography{hrruns}

\end{document}